\documentclass[a4paper,twocolumn,notitlepage,nofootinbib,superscriptaddress]{revtex4-2}
\usepackage{amsmath}
\usepackage{amssymb}
\usepackage{amsthm}
\usepackage{graphicx}
\usepackage{xcolor}
\usepackage[english]{babel}
\usepackage{csquotes}
\usepackage{braket}
\usepackage{hyperref}
\usepackage{mathtools}
\usepackage{booktabs}
\usepackage{enumitem}
\usepackage{verbatim}
\usepackage[capitalize]{cleveref}
\usepackage{soul}

\newcommand{\diff}{\ensuremath\mathrm{d}}

\newcommand{\probP}{\text{I\kern-0.15em P}}

\newcommand*{\pipe}{\ensuremath{\, | \, }}
\newcommand*{\fatpipe}{\ensuremath{\, \| \, }}

\newcommand{\Tr}{\ensuremath{\operatorname{Tr}}}

\providecommand{\myvec}[1]{\ensuremath{\boldsymbol{#1}}}

\providecommand{\bb}{\ensuremath{\myvec{b}}}

\providecommand{\xx}{\ensuremath{\myvec{x}}}

\providecommand{\ttheta}{\ensuremath{\myvec{\theta}}}

\providecommand{\calD}{\ensuremath{\mathcal{D}}}

\providecommand{\calG}{\ensuremath{\mathcal{G}}}

\providecommand{\calM}{\ensuremath{\mathcal{M}}}
\providecommand{\calN}{\ensuremath{\mathcal{N}}}

\providecommand{\calU}{\ensuremath{\mathcal{U}}}
\providecommand{\calV}{\ensuremath{\mathcal{V}}}

\providecommand{\calX}{\ensuremath{\mathcal{X}}}
\providecommand{\calY}{\ensuremath{\mathcal{Y}}}

\providecommand{\bbE}{\ensuremath{\mathbb{E}}}

\providecommand{\bbI}{\ensuremath{\mathbb{I}}}

\providecommand{\bbP}{\ensuremath{\mathbb{P}}}

\providecommand{\bbR}{\ensuremath{\mathbb{R}}}

\AtBeginDocument{
\heavyrulewidth=.08em
\lightrulewidth=.05em
\cmidrulewidth=.03em
\belowrulesep=.65ex
\belowbottomsep=0pt
\aboverulesep=.4ex
\abovetopsep=0pt
\cmidrulesep=\doublerulesep
\cmidrulekern=.5em
\defaultaddspace=.5em
}

\newtheorem{theorem}{Theorem}
\newtheorem{lemma}[theorem]{Lemma}

\newtheorem{definition}[theorem]{Definition}

\newcommand{\dccqs}{Dahlem Center for Complex Quantum Systems, Freie Universit{\"a}t Berlin, 14195 Berlin, Germany}
\newcommand{\hzb}{Helmholtz-Zentrum Berlin f{\"u}r Materialien und Energie, 14109 Berlin, Germany}
\newcommand{\hhi}{Fraunhofer Heinrich Hertz Institute, 10587 Berlin, Germany}
\newcommand{\icfo}{ICFO-Institut de Ciencies Fotoniques, The Barcelona Institute of Science and Technology, 08860 Castelldefels (Barcelona), Spain}
\newcommand{\eurecat}{Eurecat, Centre Tecnològic de Catalunya, Multimedia Technologies, 08005 Barcelona, Spain}

\begin{document}	
    \title{Single-shot quantum machine learning}
    \date{\today}
    
    \author{Erik Recio-Armengol}
    \affiliation{\icfo}
    \affiliation{\eurecat}
    
    \author{Jens Eisert}
    \affiliation{\dccqs}
    \affiliation{\hhi}
    \affiliation{\hzb}
    
    \author{Johannes~Jakob~Meyer}
    \affiliation{\dccqs}

    \begin{abstract}
        Quantum machine learning aims to improve learning methods through the use of quantum computers. 
        If it is to ever realize its potential, many obstacles need to be overcome.
        A particularly pressing one arises at the prediction stage because the outputs of quantum learning models are inherently random. This creates an often considerable overhead, as many executions of a quantum learning model have to be aggregated to obtain an actual prediction.
        In this work, we analyze when quantum learning models can evade this issue and produce predictions in a near-deterministic way -- paving the way to single-shot quantum machine learning.
        We give a rigorous definition of single-shotness in quantum classifiers and show that the degree to which a quantum learning model is near-deterministic is constrained by the distinguishability of the embedded quantum states used in the model.
        Opening the black box of the embedding, we show that if the embedding is realized by quantum circuits, a certain depth is necessary for single-shotness to be even possible.
        We conclude by showing that quantum learning models cannot be single-shot in a generic way and trainable at the same time.
    \end{abstract}

    \maketitle
	
Machine learning is a burgeoning field where rapid advances regularly overturn assumptions on what can and cannot be learned by classical computers. This ongoing success story spurs interest in \emph{quantum machine learning}, its intersection with quantum computing, another field that has seen tremendous technical progress recently. 
Investigating if quantum computers can be used to construct learning models that somehow outperform their classical counterparts has become one of the principal avenues of research in quantum machine learning. 

While the intrinsic quantum nature of such models gives them at least some theoretical potential to push beyond the boundary of what is classically possible~\cite{sweke2021quantum,liu2021rigorous,pirnay2023superpolynomial,gyurik2023exponential}, it also comes with inherent downsides. 
Previous research in that direction mostly focused on the training stage, where issues like barren plateaus~\cite{mcclean2018barren} complicate the optimization of quantum learning models. Issues do, however, also appear at the inference stage, i.e.\ when a prediction is to be produced by a quantum learning model.
A quantum learning model that deserves that name needs to perform some sort of manipulation of a quantum system. To extract a classical label, however, a measurement needs to be performed. It is the nature of quantum mechanics that the outcomes of such measurements are inherently probabilistic. For the model to solve a learning problem, for example to correctly classify an image, the output needs to be deterministic.
In practice, a quantum learning model is run a large number of times to circumvent the probabilistic outcomes of measurements and to extract a prediction, e.g.\ in the form of an expectation value of a suitably chosen observable. This issue, an instance of what has been called the \emph{measurement problem}, even persists if the model is run on a fault-tolerant quantum computer. 

\begin{figure}[h!]
    \centering
    \includegraphics{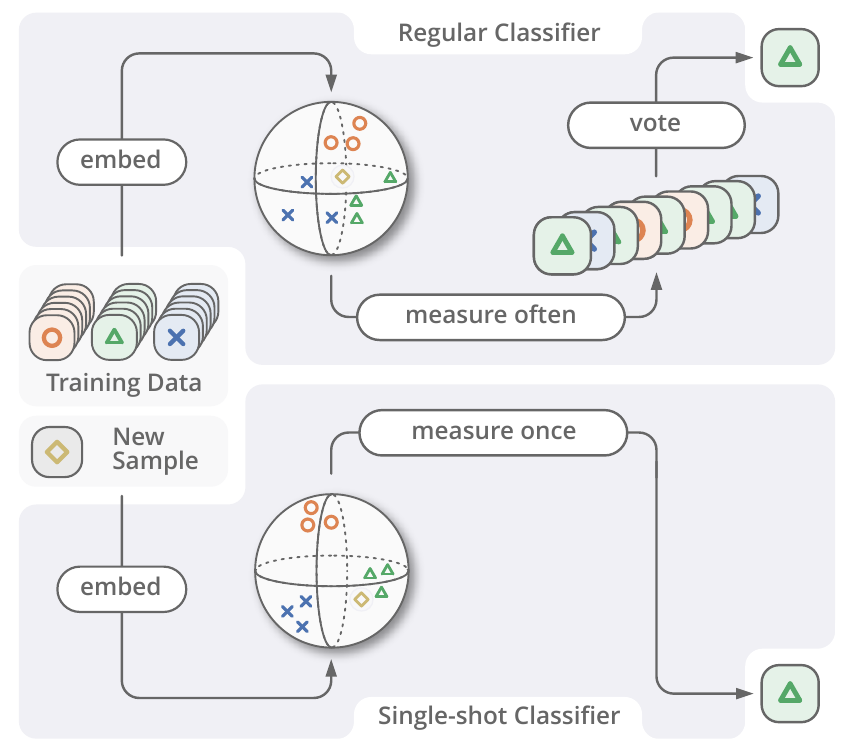}
    \caption{A depiction of the intuitive difference between a regular quantum classifier (top panel) and a single-shot one (bottom panel). In both cases, data is embedded into a quantum system through a quantum feature map. In a regular classifier, the procedure that extracts the label has to be repeated often and then aggregated into a prediction by a majority vote. In a single-shot classifier on the other hand, a single pass through the quantum model is sufficient to extract the label near-deterministically.}
    \label{fig:regular_vs_single_shot}
\end{figure}

The measurement problem means that quantum learning models usually come with overheads that are not present in classical learning models. This limits their appeal in a stark way, especially since this not only increases the \emph{time} it takes to produce a prediction but also the \emph{cost}.
In this work, we study if this particular limitation can be circumvented, specifically under which conditions it is possible to extract a prediction from a single run of a quantum learning model (see \cref{fig:regular_vs_single_shot}). We establish a rigorous definition of such \emph{single-shot models} and use tools from quantum hypothesis testing to understand their ultimate limitations.
If a quantum classifier is constructed using a quantum circuit, we establish that \emph{single-shotness} can only be possible if the circuit has a certain depth.
We treat both the noise-free and the noisy case. We complement our result by establishing that learning models that exhibit the single-shot property generically are overly expressive and therefore hard to train.

To comprehend the full potential of quantum learning models, efforts have been put into studying their potential and limitations. The fact that quantum learning models, in particular quantum classifiers, can be analyzed through the lens of multi-hypothesis testing has been put to use before in Refs.~\cite{cheng2015learnability,arunachalam2017optimal,lloyd2020quantum,weber2021optimal,du2023problem-dependent}. Ref.~\cite{du2023problem-dependent} remarks that there is a trade-off between expressivity and generalization performance. Reducing the measurement overhead of a quantum learning model has already been attempted by trying to reduce the variance of observables through training~\cite{kreplin2023reduction}.

This work is structured as follows: We start the work with expositions on quantum classifiers (\cref{sec:quantum_classifiers}) and quantum multi-hypothesis testing (\cref{sct:multi-hyp}). We go on to introduce the concept of a single-shot quantum learning model in \cref{sct:single_shot_qml}. In the subsequent \cref{sct:bounds}, we use the described concepts and tools to come up with bounds that establish that shallow models cannot have the single-shot property. Finally, in \cref{sect:hardness}, we show that a quantum learning model that achieves a generic single-shot property is hard to learn. We conclude the paper in \cref{sec:conclusion}.

\section{Quantum classifiers}\label{sec:quantum_classifiers}

A classification task can be formalized by considering data $\xx$ from a data domain $\calX$ with labels $y$ from a label domain $\calY$ which we assume to be discrete. Note that the data domain should not be seen as the ambient space, which is usually $\bbR^d$, but the subset of it which corresponds to actually valid datapoints.
In the framework of statistical learning theory, we assume that there is a joint distribution of data and labels $p(\xx,y)$ from which we obtain samples to train our classifier. We will also use the marginal distribution of the data alone, $p(\xx) = \sum_{y \in \calY} p(\xx, y)$.

A quantum classifier consists of a \emph{quantum feature map} $\xx \mapsto \rho(\xx)$ that embeds the data into the Hilbert space of a quantum system, and a procedure to extract predictions for the labels from the embedded states (compare \cref{fig:regular_vs_single_shot}). The feature map is usually implemented through a parametrized quantum circuit. 

The procedure that assigns labels to embedded data points consists of measurements and arbitrary post-processing. A particular example that is encountered often in the literature is to measure the expectation value of an observable, e.g.\ a Pauli-Z operator, and then decide for a label based on the expectation value.
On a formal level, we can absorb both the measurements and the post-processing into a POVM whose elements correspond to the different labels, \smash{$\{ \Pi_y \}_{y \in \calY}$}. 
Note that both the embedding as well as the post-processing can be trainable, in the sense that they depend on further parameters $\ttheta$, for example the angles of a parametrized quantum circuit so that $\rho(\xx) = \rho_{\ttheta}(\xx) = \calU(\ttheta, \xx)[\rho_0]$.

A quantum classifier is a particular instance of what we refer to as a \emph{probabilistic classifier}, as it -- contrary to for example a classical neural network -- outputs labels non-deterministically. Abstractly speaking, the output of the classifier $f$ is not a label, but a random variable \smash{$\hat{f}(\xx)$} with values in $\calY$ and distribution \smash{$\bbP[\hat{f}(\xx) = y] = \Tr[ \Pi_y \rho(\xx) ]$}. 
The probabilistic nature is problematic, because we need to assign a unique label to every datapoint. In quantum machine learning, this problem is circumvented by running the classification procedure multiple times and taking a majority vote to predict the most likely label
\begin{align}\label{eqn:def_fx}
    f(\xx) \coloneqq\operatornamewithlimits{argmax}_{y'} \Tr[ \rho(\xx) \Pi_{y'}].
\end{align}
This is also implicitly happening when the expectation value of an observable is used to determine a label, as multiple runs of a quantum experiment are usually needed to approximate the expectation value faithfully. 

\section{Quantum multi-hypothesis testing}\label{sct:multi-hyp}

Quantum classifiers solve a task that is similar in spirit to \emph{quantum multi-hypothesis testing}. This refers to the task of distinguishing quantum states from a set of known states \smash{$\{ \rho_j \}_{j=1}^r$}~\cite{khatri2020principles}. It is an important tool used in quantum information theory because many tasks that seem unrelated to hypothesis testing a priori can be reduced to it. Bounds on the best possible achievable precision in the multi-hypothesis testing task then yield corresponding bounds for the initial task. As we will see shortly, quantum classifiers are no exception to this.

In multi-hypothesis testing, our job is to design a measurement that upon input of the state $\rho_j$ produces the output $j$ with high probability. We can model the most general quantum measurement by a positive operator-valued measurement (POVM) $\{ \Pi_j \}_{j=1}^M$ whose operators satisfy $\Pi_j \geq 0$ and \smash{$\sum_{j=1}^M \Pi_j = \bbI$}.
In the Bayesian version of quantum multi-hypothesis testing, we are given not only the quantum states \smash{$\{ \rho_j \}_{j=1}^r$}, but also the probabilities with which we are given the state, \smash{$\{p_j\}_{j=1}^r$}. Our job is to find a measurement that minimizes the Bayesian multi-hypothesis testing error
\begin{align}\label{eqn:def_multi_hyp_error_bayesian}
    P_e^{*}( \{ p_j \rho_j \}_{j=1}^r) &= \min_{\substack{\{ \Pi_j \} \\ \mathrm{POVM}} } \left\{ 1-\sum_{j=1}^{r} p_j\Tr[\Pi_{j}\rho_{j}]\right\}.
\end{align}

If we do not have knowledge about the prior probabilities of the different states, we can still ask for a worst-case guarantee. This yields the minimax hypothesis testing error which can be guaranteed for any prior distribution of the states:
\begin{align}\label{eq:agnostic_error}
    \overline{P}_e^{*}( \{ \rho_j \}_{j=1}^r) &= \min_{\substack{\{ \Pi_j \} \\ \mathrm{POVM}} } \max_{1 \leq j \leq r} \left\{ 1-\Tr[\Pi_{j}\rho_{j}]\right\}.
\end{align}
The optimal measurement in both the Bayesian and the minimax setting can be computed with a semi-definite program (SDP)~\cite{yuen1975optimum}.

The usefulness of hypothesis testing stems from the fact that we have good lower bounds on the multi-hypothesis testing error. In the case of binary hypothesis testing, $r=2$, we even have a closed form solution.
The Helstrom-Holevo Theorem \cite{helstrom1969quantum} relates the minimal error of binary hypothesis testing to the trace distance of the quantum states. 
\begin{align}\label{helstrom}
    P_e^{*}(p \rho, (1-p) \sigma) = \frac{1}{2} - \frac{1}{2}\lVert p \rho - (1-p) \sigma \rVert_1.
\end{align}
We immediately see that the above expression vanishes for orthogonal states and, assuming without loss of generality that $p\geq 1/2$, takes its maximal value $1 - p$ for $\rho = \sigma$, where $1-p$ amounts to the probability of success for random guessing.

For more than two states, we do not have closed-form solutions. We can, however, still retain the interpretability in terms of pairwise distances given by the Holevo-Helstrom theorem through the following lemma which is inspired by a similar construction in Ref.~\cite{Audenaert_2014}.
\begin{lemma}[Multi-hypothesis testing to binary reduction]\label{lemma_reduction}
    For a given quantum multi-hypothesis testing problem between states $\{ p_i \rho_i \}_{i=1}^r$ we have that
    \begin{align}
    \begin{split}
        P_e^{*}(& \{p_i  \rho_i \}_{i=1}^r ) \\
        & \geq \sum_{i=1}^r p_i \max_{1 \leq j \neq i \leq r} P_e^{*}\left(\frac{p_i}{p_i + p_j} \rho_i, \frac{p_j}{p_i + p_j} \rho_j\right)
    \end{split}\\
        & \geq \min_{1 \leq i \leq r} \max_{1 \leq j \neq i \leq r} P_e^{*}\left(\frac{p_i}{p_i + p_j} \rho_i, \frac{p_j}{p_i + p_j} \rho_j\right).
    \end{align}
\end{lemma}
The proof is presented in \cref{appendix_multi_ht_to_binary}. The corresponding result in the minimax setting is straightforward, as the worst-case error can only improve when considering a hypothesis test among fewer alternatives, hence
\begin{align}
    \overline{P}_e^*(\{ \rho_j \}_{j=1}^r ) \geq \max_{i \neq j} \overline{P}_e^*(\rho_i, \rho_j).
\end{align}
Combining the above results with the Holevo-Helstrom theorem, we obtain a lower bound on the error for the multi-hypothesis testing problem in terms of the pair-wise distances of the quantum states. We note that these lower bounds are not necessarily the tightest possible. But they allows us to later connect the limitations on single-shot quantum learning models to interpretable quantities based on the trace distance of quantum states.

\subsection*{Reduction to multi-hypothesis testing for quantum classifiers}\label{section_reduction}

This section is devoted to show that quantum classifiers can be naturally analyzed through the lens of quantum multi-hypothesis testing. For similarly inspired derivations, see also Refs.~\cite{cheng2015learnability,lloyd2020quantum,weber2021optimal,du2023problem-dependent}. 
To do so, we first have a look at the probability that a given quantum classifier outputs the correct label when it is used once. This is nothing but the expected accuracy of the classifier. In the following calculation, we will split the parameter space according to labels, i.e.\ we define $\calX_{y'} \coloneqq \{ \xx \pipe y(\xx) = y' \}$.
\begin{align}
    \bbP[\text{success}] &= \int_\calX \diff p(\xx, y) \, \bbP[ \hat{f}(\xx) = y] \label{eq_reduction_succ} \\ 
        &= \sum_{y' \in \calY} \int_{\calX_{y'}} \diff p(\xx) \, \Tr[ \Pi_{y'} \rho(\xx) ] \\
        &= \sum_{y' \in \calY} \Tr\left[ \Pi_{y'} \int_{\calX_{y'}} \diff p(\xx) \, \rho(\xx) \right] \\
        &= \sum_{y' \in \calY} \Tr\left[ \Pi_{y'} p_{y'} \rho_{y'}\right].
\end{align}
In the last step, we defined the average state for every class
\begin{align}
    p_{y'} &\coloneqq \int_{\calX_{y'}} \diff p(\xx) \\
   \rho_{y'}&\coloneqq \frac{1}{p_{y'}} \int_{\calX_{y'}} \diff p(\xx) \, \rho(\xx).
\end{align}
If we compare the above with \cref{eqn:def_multi_hyp_error_bayesian}, we see that a quantum classifier is equivalent to a quantum multi-hypothesis testing problem for the average states for each class.
Therefore, we can conclude that the error of a quantum classifier is at least the optimal achievable error
\begin{align}
    \bbP[\text{error}] \geq P_e^{*}( \{ p_{y} \rho_{y} \}_{y \in \calY}).
\end{align}
The fact that the embedded states for different classes have to be well-distinguishable to achieve reasonable accuracy -- also measured in terms of other loss functions -- has already been observed in prior works~\cite{cheng2015learnability,lloyd2020quantum,du2023problem-dependent}.

\section{Single-shot quantum machine learning} \label{sct:single_shot_qml}

The fact that we have to run a quantum classifier multiple times, explained in \cref{sec:quantum_classifiers}, leads to the so-called \enquote{measurement problem}, meaning that there are additional overheads compared to classical approaches which stem from the probabilistic nature of quantum mechanics. The purpose of this work is to formalize when we can circumvent this issue by enforcing that the classifier is so sure about the assigned label that running it only once is sufficient, which is when we call the classifier \enquote{single-shot}. 

\subsection*{Single-shot classifiers}

With the notation we introduced, we can define when a classifier is single-shot: exactly when the probability that a random label $\hat{f}(\xx)$ is identical to the majority label $f(\xx)$ assigned by the classifier is very high. Formally, we define
\begin{definition}[Single-shot probabilistic classifier (Bayesian)]\label{def:bayesian_single_shot_classifier}
    For datapoints $\xx$ distributed according to a distribution $p(\xx)$ supported on a data domain $\calX$, we say a probabilistic classifier $f$ is $\delta$-single shot if
    \begin{align}
    \operatornamewithlimits{\mathbb{E}}_{\xx \sim p(\xx)} \left\{ \bbP[ \hat{f}(\xx) = f(\xx)] \right\} \geq 1 - \delta,
    \end{align}
    where the expectation value is taken over the random distribution of datapoints.
\end{definition}
Note that $\bbP[ \hat{f}(\xx) = f(\xx)] = \max_{y \in \calY} \Tr[ \Pi_{y} \rho(\xx) ]$, which means we can rewrite this as
\begin{align}\label{eqn:bayesian_single_shot_condition_rewrite}
    \int \diff p(\xx) \, \max_{y \in \calY} \Tr[ \Pi_{y} \rho(\xx) ] \geq 1 - \delta.
\end{align}

It is worth to spend some time to clarify multiple important points: 
First, being \emph{single-shot} is a property of the classifier that is independent of the labels being actually correct. The single-shot property only captures how \emph{deterministic} the outcomes of a classifier are. Consider for example a classifier that always outputs the same label, it will be perfectly single-shot, i.e.\ have $\delta = 0$, but will have bad classification accuracy.
Second, the single-shot property is necessarily dependent on the learning problem through the data distribution $p(\xx)$: a classifier that has this property for one problem will usually not have it for a different learning problem.
Third, how strict the above single-shot definition actually is crucially depends on the allowed failure probability $\delta$, as \emph{every} classifier will have this property if the failure probability $\delta$ can be chosen arbitrarily close to $1$. We are, naturally, interested in the case where $\delta$ is rather small.
Fourth, we wish to emphasize that being single-shot is a property that we wish to have for a \emph{trained} classifier. It is no property that will be present upon random initialization and has to be enforced through an appropriate training procedure in one way or another.

We refer to \cref{def:bayesian_single_shot_classifier} as the \emph{Bayesian} single-shot definition because we average over the distribution of the inputs $\xx \sim p(\xx)$. We can obtain a more stringent definition that removes this dependence by minimizing over all possible distributions $p(\xx)$ supported on the data domain $\calX$. Here, again, we emphasize that the data domain should be understood not as the whole ambient space of the data (usually $\bbR^d$), but the subset corresponding to valid inputs.
In this case, it is intuitively clear that the \enquote{worst} distribution is the one that always gives our classifier the datapoints for which \smash{$\bbP[ \hat{f}(\xx) = f(\xx) ]$} is minimal. This leads us to the following \emph{agnostic} definition of being single-shot:
\begin{definition}[Single-shot probabilistic classifier (Agnostic)]\label{def:agnostic_single_shot_classifier}
    For datapoints $\xx$ from a data domain $\calX$, we say a probabilistic classifier $f$ is agnostically $\bar\delta$-single-shot if
    \begin{align}
        \min_{\xx \in \calX} \bbP[ \hat{f}(\xx) = f(\xx) ] \geq 1 - \bar\delta,
    \end{align}
    where the probability is taken over the inherent randomness of the classifier.
\end{definition}
Analogously to \cref{eqn:bayesian_single_shot_condition_rewrite}, we can rewrite this as 
\begin{align}
    \min_{\xx \in \calX} \max_{y} \Tr[ \Pi_{y} \rho(\xx) ] \geq 1 - \bar\delta.
\end{align}
The agnostic definition immediately implies the Bayesian one with the same $\delta = \bar\delta$, but it is much more stringent. This is especially clear if we assume the classifier to be continuous in the data, in which case agnostic single-shot classification is not possible with a meaningfully small value of $\bar\delta$ if there are two sets of points with distinct labels assigned by the classifier that touch in parameter space as continuity would imply that at the point where the label flips, both labels are equally likely and thus single-shot classification is not possible.

It is, however, quite intuitive that even a continuously parametrized classifier can be agnostically single-shot on subsets of the data domain, and if the points we see are likely to come from that domain, we still have a pretty good Bayesian single-shot performance. Exactly this thought is formalized by the following lemma, which is just a simple application of the union bound.
\begin{lemma}[Restriction lemma]\label{lem:agnostic_to_bayesian_lifting}
    Let $\hat{f}(\xx)$ be a probabilistic classifier that is agnostically $\bar\delta$-single-shot relative to a subset of the data domain $\bar\calX \subseteq \calX$. Then, for any data-distribution, the classifier is Bayesian $\delta$-single-shot with
    \begin{align}
        \delta \leq \bar\delta + \bbP[ \xx \not\in \bar\calX \pipe \xx \sim p(\xx)].
    \end{align}
\end{lemma}

\subsection*{Geometric Picture}

It is instructive to visualize the concept of single-shot quantum classification. 
To do so, we can define some specific sets of inputs. 
First, we define the set of all inputs that the classifier assigns the same majority label
\begin{align}\label{eqn:def_parameter_space_for_classifier_labels}
    \tilde{\calX}_y \coloneqq \{ \xx \in \calX \pipe f(\xx) = y \}.
\end{align}
Next, for a given failure probability $\delta$, we define the subset of the above set on which the classifier is agnostically $\delta$-single-shot according to \cref{def:agnostic_single_shot_classifier}.
\begin{align}
\begin{split}
    \tilde{\calX}^{\delta}_y \coloneqq \{ \xx \in &\calX \pipe f(\xx) = y, \\ &\bbP[\hat{f}(\xx) = f(\xx)\pipe \xx] \geq 1- \delta \}.
\end{split}
\end{align}
It is quite natural that $\delta \leq \delta'$ implies that \smash{$\tilde{\calX}^{\delta} \subseteq \tilde{\calX}^{\delta'}$}. Plotting the sets \smash{$\tilde{\calX}_y$} amounts to plotting the decision boundaries of the classifier, whereas plotting the sets \smash{$\tilde{\calX}_y^{\delta}$} shows the regions in which the classifier is additionally \enquote{sure} about its decisions.

Naturally, we also care about the true labels of the datapoints, as such we will use the following notation to denote all datapoints with a given label:
\begin{align}
    \calX_y \coloneqq \{ \xx \in \calX \pipe y(\xx) = y \}.
\end{align}
\cref{fig:single_shot_sets} shows an exemplary plot of some sets for a trivalent classification problem for points in $\bbR^2$ that clarifies the nature of the sets $\tilde{\calX}_y^{\delta}$, $\tilde{\calX}_y$ and $\calX_y$.

\begin{figure}
    \centering
    \includegraphics{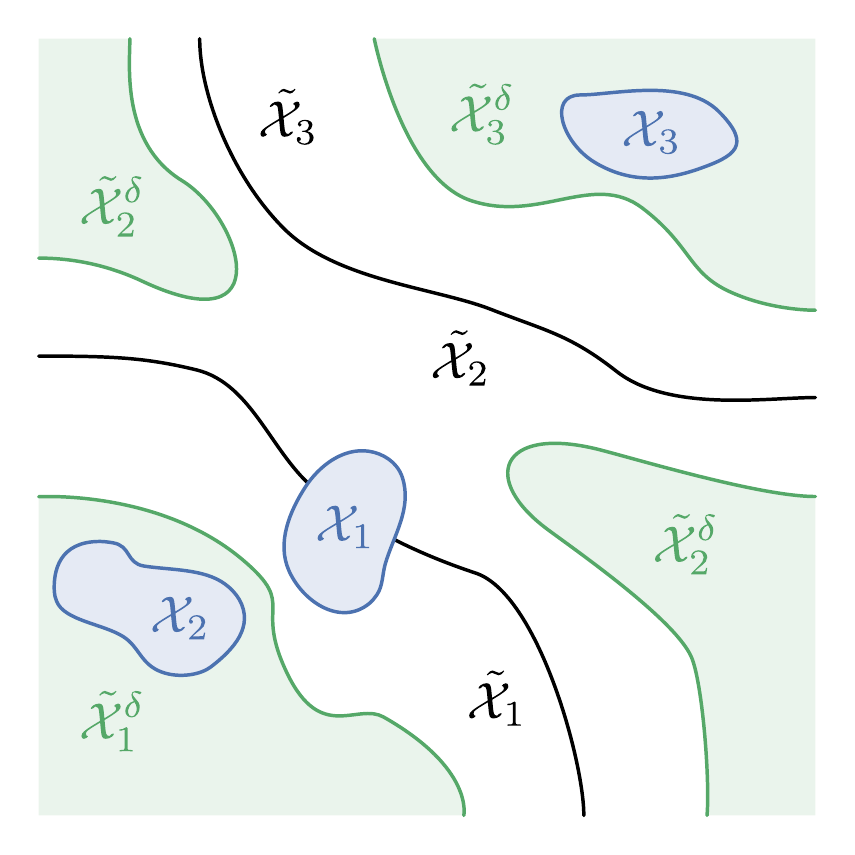}
    \caption{To gain intuition about the notion of single-shot classification, we look at a problem where points in $\calX = \bbR^2$ are classified into three classes. The sets \smash{$\tilde{\calX}_i$} correspond to the different labels assigned by the classifier (boundaries in black). A subset of each of these sets, \smash{$\tilde{\calX}_i^{\delta}$} corresponds to the inputs for which these labels are assigned in a $\delta$-single-shot way (green). The supports of the datapoints with true labels are given by $\calX_i$ (blue).
    We observe three different situations: Data from class 1 is not always correctly classified and labels are not assigned in a $\delta$-single-shot way. Data from class 2 is incorrectly classified, but labels are assigned in a single-shot way. Data from class 3 is correctly classified in a single-shot way.}
    \label{fig:single_shot_sets}
\end{figure}

\subsection*{Relation to Sample Complexity}

We can relate our notion of a probabilistic classifier being single-shot back to the notion of sample complexity that we usually encounter when dealing with the measurement problem in quantum machine learning models. This can be achieved by considering that the process of using $m$ repetitions of the same \enquote{weak} classifier followed by a majority vote can be viewed as a single model that should then have the single-shot property, which amounts to a low failure probability. 
Mathematically speaking, we have the following lemma about boosting the success probability of a weak classifier, i.e.\ one that has the single-shot property only with a high failure probability. 
\begin{lemma}\label{lem:majority_vote_boosting}
    Let $f$ be a probabilistic classifier on a data domain $\calX$. Then, $m$ independent repetitions of $f$ together with a majority vote form a new classifier that is agnostically $\bar\delta$-single-shot with
    \begin{align}
        \bar\delta \leq |\calY| \exp\left(- \frac{1}{2} m \Delta^2\right),
    \end{align}
    where $\Delta$ is the smallest gap between the two largest label probabilities $\Tr[ \Pi_{y} \rho(\xx) ]$ for all $\xx$.
\end{lemma}
\begin{proof}
We define the random variables $\{ \hat{r}_{y}(\xx) \}_{y \in \calY}$ as the ratio of the $m$ classifiers that output the label $y$, i.e.
\begin{align}
    \hat{r}_y(\xx) = \frac{1}{m} \sum_{i=1}^m \chi[ \hat{f}_i(\xx) = y ],
\end{align}
understanding that $\bbE[ \hat{r}_y(\xx) ] = \Tr[ \Pi_{y} \rho(\xx) ]$. From Hoeffding's inequality, we have the two one-sided estimates
\begin{align}
    \bbP[ \hat{r}_y(\xx) - \Tr[ \Pi_{y} \rho(\xx) ] \geq \epsilon] &\leq \exp( - 2 m \epsilon^2 ) \\
    \bbP[ \Tr[ \Pi_{y} \rho(\xx) ] - \hat{r}_y(\xx) \geq \epsilon] &\leq \exp( - 2 m \epsilon^2 ).
\end{align}
We could have used a tighter estimate deriving from the Chernoff-Hoeffding Theorem on the sum of i.i.d.\ Bernoulli random variables, but the improvements are not necessary for our argument.
We can guarantee a successful outcome of the majority vote, if we choose $\epsilon$ to be half the gap between the largest and the second-largest value of $\Tr[ \Pi_{y} \rho(\xx) ]$. Formally, we can define the gap as
\begin{align}
    \Delta(\xx) \coloneqq \min_{\calY \ni y' \neq y} \max_{y \in \calY} \Tr[ (\Pi_{y} - \Pi_{y'}) \rho(\xx) ],
\end{align}
and hence choose $\epsilon = \epsilon(\xx) = \Delta(\xx)/2$. We have to use Hoeffding's inequality both to assure that $\hat{r}_y(\xx)$ associated to the largest value is not too low, as well that all of the other options are not too large. Combining them through the union bound leads us to the bound
\begin{align}
    \bar\delta(\xx) \leq |\calY| \exp\left(-{\frac{1}{2}m \Delta(\xx)^2}\right).
\end{align}
Taking the worst case over $\xx$ and defining
\begin{align}
    \Delta = \min_{\xx \in \calX} \Delta(\xx)
\end{align}
yields the statement of the lemma.
\end{proof}
We wish to note that the dependence on the gap instead of the failure probability of the classifier is necessary. We can connect them and obtain a lower bound on the gap $\Delta \geq 2\epsilon$ if the classifier used in the majority voting procedure has failure probability $\bar\delta = \frac{1}{2}+\epsilon$.

A Bayesian version of the argument of \cref{lem:majority_vote_boosting} proceeds similarly, but the resulting right hand side is
\begin{align}
    \delta \leq |\calY|\int\diff p(\xx) \, \exp\left(-\frac{1}{2}m \Delta(\xx)^2 \right),
\end{align}
which is not as easy to interpret and to rearrange to obtain a lower bound on $m$ for desired $\delta(\xx)$, which is why we think it better to use the agnostic version of \cref{lem:majority_vote_boosting} together with the restriction argument of \cref{lem:agnostic_to_bayesian_lifting} to arrive at a Bayesian version of the result. There, the restriction of the parameter space $\bar\calX$ should be chosen such that a lower bound on the gap $\Delta$ can be guaranteed for all $\xx \in \bar\calX$.

\section{Ultimate limits of single-shot quantum machine learning} \label{sct:bounds}

Having defined the single-shot property, it is natural to ask if it can actually be achieved for a given embedding $\rho(\xx)$ and, vice-versa, what resources are necessary to realize an embedding that allows for single-shot classification. This section is devoted to exploring the ultimate limits of single-shot quantum classifiers by reducing the single-shot classification task to a multi-hypothesis testing task. After establishing said reduction, we open the black box of the embedding and establish lower bounds on the gate complexity a quantum circuit needs to establish the single-shot property.

\subsection*{Reduction to hypothesis testing}

In \cref{section_reduction}, we have seen that the error probability of quantum classifiers can be lower-bounded through a reduction to multi-hypothesis testing. 
We can employ a similar argument to bound the best value of the single-shot error probability $\delta$. The subtle difference lies in the fact that we are no longer concerned with classifying the data well with respect to the true labels $y(\xx)$ as they don't figure in the definition of the single-shot property. Instead, we want the classifier to work with respect to the labels \emph{assigned by the classifier itself}, i.e.\ $f(\xx)$ as in \cref{eqn:def_fx}. After this change of perspective the argument can, however, proceed analogously.

To quantify the limits of the Bayesian single-shot property, we split the parameter space according to the majority labels assigned by the classifier and recall the definition \smash{$\tilde{\calX}_{y'} \coloneqq \{ \xx \in \calX \pipe f(\xx) = y' \}$} from \cref{eqn:def_parameter_space_for_classifier_labels}. We obtain the following theorem.
\begin{theorem}[Bayesian error probability lower bound]\label{thm:bayesian_error_probability_lower_bound}
Let $f$ be a probabilistic classifier on a data domain $\calX$ taking discrete values in $\calY$. For datapoints $\xx$ distributed according to a distribution $p(\xx)$, we define the average states for the classes assigned by $f$ as
\begin{align}
    \tilde{p}(y) &\coloneqq \int_{\tilde{\calX}_y} \diff p(\xx), \ \,
    \tilde{\rho}_y \coloneqq \frac{1}{\tilde{p}(y)}\int_{\tilde{\calX}_y} \diff p(\xx)\, \rho(\xx).
\end{align}
If the classifier $f$ has the Bayesian single-shot property, the error probability $\delta$ has to fulfill
\begin{align}
    \delta \geq P_e^{*}( \{ \tilde{p}_{y} \tilde{\rho}_y \}_{y \in \calY}).
\end{align}
\end{theorem}
\begin{proof}
We start from \cref{def:bayesian_single_shot_classifier} and obtain
\begin{align}
    1-\delta &\leq \operatornamewithlimits{\mathbb{E}}_{\xx \sim p(\xx)} \left\{ \bbP[ \hat{f}(\xx) = f(\xx)] \right\}\\
        &= \int_\calX \diff p(\xx) \, \bbP[ \hat{f}(\xx) = f(\xx)] \\ 
        &= \int_{\calX} \diff p(\xx) \,\max_{y' \in \calY}\Tr[ \Pi_{y'} \rho(\xx) ] \\
        &= \sum_{y' \in \calY} \int_{\tilde{\calX}_{y'}} \diff p(\xx) \,  \Tr[ \Pi_{y'} \rho(\xx) ] \\
        &= \sum_{y' \in \calY} \Tr\left[ \Pi_{y'} \int_{\tilde{\calX}_{y'}} \diff p(\xx) \, \rho(\xx) \right] \\
        &= \sum_{y' \in \calY} \Tr\left[ \Pi_{y'} \tilde{p}_{y'} \tilde{\rho}_{y'}\right] \\
        &\leq P_s^{*}(\{ \tilde{p}_{y} \tilde{\rho}_y \}_{y \in \calY}).
\end{align}
With $P_s^{*}(\{ \tilde{p}_{y} \tilde{\rho}_y \}_{y \in \calY})$ being the optimal Bayesian
multi-hypothesis testing success probability. Rearranging and using $P_s^{*} = 1 - P_e^{*}$ yields the statement of the theorem.
\end{proof}
In the above theorem, the data space $\calX$ and the marginal distribution of the data $p(\xx)$ are the same as in \cref{section_reduction}, only the average state $\tilde{\rho}_{y'}$ is now taken among the states which are assigned the same label by the classifier and is not related to the actual correct labels of the dataset. The lower bound on $\delta$ of the above Theorem is achieved by the optimal measurement for the Bayesian multi-hypothesis testing problem between the states $\tilde{\rho}_y$ with prior probabilities $\tilde{p}_y$.

In the agnostic setting, the single-shot property has to hold irrespective of the underlying distribution of the data. As such, we can bound it by constructing the \enquote{worst} possible distribution over the data, i.e.\ the one that minimizes the Bayesian success probability. It is clear that this is distribution consists of putting all the probability on the single datapoint $\xx$ for which the classifier struggles the most.\footnote{It is important to emphasize that in this scenario, the adversarial distribution over the data is chosen \emph{after} a POVM performing the classification was fixed, if it were the other way round a different strategy of the adversary would be optimal.}
This strategy yields the following bound.
\begin{theorem}[Agnostic error probability lower bound]\label{thm:agnostic_error_probability_lower_bound}
Let $f$ be a probabilistic classifier on a data domain $\calX$ taking discrete values in $\calY$.
If the classifier $f$ has the agnostic single-shot property, the error probability $\bar\delta$ has to fulfill
\begin{align}
    \bar\delta \geq \max_{ \{\xx_y \in \tilde{\calX}_y \}_{y \in \calY} } \overline{P}_e^{*}( \{ \rho(\xx_y) \}_{y \in \calY}).
\end{align}
\end{theorem}
\begin{proof}
    With the reasoning described above, we see that for the adversarially chosen distribution we have
    \begin{align}
    1-\bar\delta &\leq \min_{\xx \in \calX} \bbP[ \hat{f}(\xx) = f(\xx)] \\
    &=\min_{\xx \in \calX} \max_{y' \in \calY} \Tr[ \Pi_{y'} \rho(\xx) ].
    \end{align}
    We can write
    \begin{align}
        \min_{\xx \in \calX}\max_{y' \in \calY} = \min_{y \in \calY} \min_{\xx_y \in \tilde{\calX}_y},
    \end{align}
    because we know by definition of \smash{$\tilde{\calX}_y$} which projector achieves the maximum over $y'$. This gives
    \begin{align}
        1 - \bar\delta &\leq \min_{y \in \calY} \min_{\xx_y \in \tilde{\calX}_y}\Tr[ \Pi_{y} \rho(\xx_y) ].
    \end{align}
    To exchange the two minima, we optimize over sets of data points \smash{$\{ \xx_y \in \tilde{\calX}_y \}_{y \in \calY}$}, one for each possible label. This gives
    \begin{align}
        1 - \bar\delta &\leq \min_{\{ \xx_y \in \tilde{\calX}_y \}_{y \in \calY}} \min_{y \in \calY}\Tr[ \Pi_{y} \rho(\xx_y) ] \\
        &\leq \min_{\{ \xx_y \in \tilde{\calX}_y \}_{y \in \calY}} \overline{P}_s^{*}( \{ \rho(\xx_y) \}_{y \in \calY}),
    \end{align}
    where we recognized the definition of the minimax success probability for multi-hypothesis testing from \cref{eq:agnostic_error} in the last step. The theorem statement follows by using $\overline{P}_s^{*} = 1 - \overline{P}_e^{*}$ and rearranging.
\end{proof}
We note that we could in theory further improve the lower bound on the agnostic single-shot error probability in the above bound by including multiple states per label which would result in a reduction to a composite multi-hypothesis testing problem. These problems are less studied and fewer analytical tools are available for them. For numerical investigations this is, however, an adequate way to achieve tighter bounds in conjunction with semi-definite programming.

We now use the tools presented in \cref{sct:multi-hyp} to connect the ultimate limits of single-shot quantum classification to more intuitive quantities.
We apply the reduction to the binary case of \cref{lemma_reduction} combined with the Holevo-Helstrom theorem of \cref{helstrom} to the Bayesian error probability lower bound of \cref{thm:bayesian_error_probability_lower_bound}. This gives the following lower bound on the error probability in terms of the trace distance of the average states for the different classes:
\begin{align}\label{eqn:bayesian_states_bound}
    \delta &\geq \sum_{y \in \calY} \tilde{p}_y \max_{y \neq y' \in \calY} \frac{1}{2} e(y,y')\\
    &\geq \min_{y \in \calY} \max_{y \neq y' \in \calY} \frac{1}{2} e(y, y'),
\end{align}    
where we defined
\begin{align}
    e(y, y') \coloneqq 1 - \bigg\| \frac{\tilde{p}_y}{\tilde{p}_y + \tilde{p}_{y'}} \tilde{\rho}_y - \frac{\tilde{p}_{y'}}{\tilde{p}_{y} + \tilde{p}_{y'}} \tilde{\rho}_{y'} \bigg\|_1.
\end{align}
This expression shows that the closer the average states of the different classes assigned by the classifier are, the worse the achievable error probability becomes and the less \enquote{single-shot} our classifier must be. Therefore, in order to have a small error probability, the average states of different assigned labels have to be far apart.

In the agnostic case, we can directly reduce to a binary discrimination problem as trying to distinguish less states in a minimax fashion will certainly be easier, i.e.\ we have that
\begin{equation}
    \overline{P}_e^{*}(\{ \rho_j \}_{j=1}^r)) \geq \max_{1 \leq i \neq j \leq r } \overline{P}_e^{*}(\rho_i, \rho_j).
\end{equation}
While there is no closed-form expression for the minimax binary hypothesis testing error probability, it is reasonable to assume that it is not too far from the problem of distinguishing both states with prior probability $1/2$ for each state. Therefore:
\begin{equation}
    \overline{P}_e^{*}(\rho_i, \rho_j) \geq P_e^{*}\left(\frac12\rho_i, \frac12 \rho_j\right).
\end{equation}

Combining this reasoning again with the Holevo-Helstrom bound and \cref{thm:agnostic_error_probability_lower_bound}, we obtain the following lower bound for the agnostic single-shot error probability
\begin{align}
    \bar\delta &\geq \max_{\substack{\xx, \xx'\\ f(\xx)\neq f(\xx')}} \overline{P}_e^{*}(\rho(\xx), \rho(\xx')) \\
    &=\max_{\substack{\xx, \xx'\\ f(\xx)\neq f(\xx')}} P_e^{*}\left(\frac12\rho(\xx), \frac12 \rho(\xx')\right)\\
    &\geq\max_{\substack{\xx, \xx'\\ f(\xx)\neq f(\xx')}} \frac{1}{2} \left(1 - \frac{1}{2}\lVert \rho(\xx) - \rho(\xx') \rVert_1 \right)\\
    &=  \frac{1}{2} - \frac{1}{4}\min_{\substack{\xx, \xx'\\ f(\xx)\neq f(\xx')}}\lVert \rho(\xx) - \rho(\xx') \rVert_1.
\end{align}
From this, we can conclude that agnostic single-shotness imposes a strong requirement on the distances of the embedded states for different classes predicted by the classifier, as we can reverse the inequality above to obtain
\begin{align}
    \frac{1}{2}\lVert \rho(\xx) - \rho(\xx') \rVert_1 \geq 1 - 2\bar\delta
\end{align}
for all $\xx$ and $\xx'$ belonging to different classes as predicted by the classifier. 

It is quite obvious from this bound that any embedding $\rho(\xx)$ that continuously varies over $\xx$ cannot be agnostically single shot on the whole value range of $\xx$. This emphasizes again that the agnostic definition of single-shotness necessitates a restriction of the domain of $\xx$, $\calX$, as we also used in the lifting from agnostic to Bayesian in \cref{lem:agnostic_to_bayesian_lifting}.

Having established a connection of the single-shot error probability in both the agnostic and the Bayesian setting to the trace distance of the embedded states, we open up the black-box of the embedding $\rho(\xx)$ and consider that it is realized by a variational quantum circuit. In that case, imposing a certain single-shot error probability necessitates a certain depth and/or width of the circuit as we show below. 

\subsection*{Noiseless quantum circuits}

We will first consider the case where the data is embedded by the means of a noiseless quantum circuit. We can express any such circuit as a unitary composed from parametrized gates $R_G(\alpha) = e^{-i G \alpha}$, where $G$ is the \emph{generator} of the gate, and fixed unitaries. In our case, we assume that the circuit is partitioned into $L$ layers, where each layer corresponds to a subcircuit that encodes every parameter $x_i$ once through a parametrized gate and can consist of arbitrary variationally parametrized gates and fixed unitaries otherwise. In \cref{appendix_continuity}, we prove that parametrized quantum circuits have a continuity bound in terms of the arguments of the parametrized gates, which, as a corollary, yields the following result:
\begin{lemma}[Continuity of parametrized quantum states]\label{lem:continuity_of_parametrized_states}
    Let $\rho_{\ttheta}(\xx)$ be obtained from a parametrized quantum circuit of $L$ layers, where in each layer the parameters \smash{$\{ x_i \}_{i=1}^d$} are encoded through separate parametrized gates. Then, we have that
    \begin{align}\label{eqn:continuity_bound}
        \lVert \rho_{\ttheta}(\xx) - \rho_{\ttheta}(\xx') \rVert_1 \leq L \Delta \lVert \xx - \xx' \rVert_1,
    \end{align}
    where $\Delta \coloneqq \max_{G \in \calG}( \lambda_{\max}(G) - \lambda_{\min}(G))$ is the largest spectral spread of the generators used to encode the parameters \smash{$\{ x_i \}_{i=1}^d$}.
\end{lemma}

\cref{lem:continuity_of_parametrized_states} shows us that a parametrized circuit embeds datapoints that are close into quantum states that are close as well. Intuitively, we can reason that the quantum classifier is going to have a hard time with two datapoints that are close but belong to different assigned labels. We formalize this in the following theorem.
\begin{theorem}[Lower bound on the single-shot error probability (Bayesian)]\label{theorem_bayesian_bound}
Let \smash{$\hat{f}(\xx)$} be a quantum classifier with the task of classifying $r$ groups of classical data for which we know its probability distribution. The single-shotness error probability is bounded by the worst average distance between two data classes:
\begin{align} \label{eqn:theorem_bayesian_bound}
    \delta \geq  \min_{1 \leq i \leq r} \max_{1 \leq j \neq i \leq r} \frac{\min(p_i,p_j)}{p_i + p_j}\left(1 - \frac{L\Delta d_\text{avg}^{ij}}{2}\right),
\end{align}
where
\begin{align}
d_\text{avg}^{ij} \coloneqq (p_i + p_j) \int_{\mathcal{X}_i}\frac{dp(\xx_i)}{p_i} \int_{\mathcal{X}_j}\frac{dp(\xx_j)}{p_j} \|\xx_i-\xx_j\|_1
\end{align}
is the expected distance of two datapoints sampled according to the data distribution, conditioned on them being from the classes $i$ and $j$.
\end{theorem}
The proof is presented in \cref{proof_continuity_noiseless}. The average distance $d_{\mathrm{avg}}^{ij}$ gives us a sense of how far apart the points of the two classes $i$ and $j$ are in the input data space. 

In this theorem we see that in order to have a low single-shot error probability, either we have that the different classes are located far apart in the input space or we have a sufficiently deep circuit. Therefore, fixing a dataset, this expression shows us that the depth of the circuit can be limiting the classifier performance. To get some intuition, we assume that all $r$ classes are equally likely, i.e. that $p = 1/r$, in which case we can rearrange the expression in \cref{eqn:theorem_bayesian_bound} to obtain: 
\begin{align}
    L \geq 2 \frac{1-2\delta}{\Delta d_{\text{avg}}},
\end{align}
where
\begin{align}
d_{\text{avg}} = \max\limits_{1 \leq i \leq r} \min\limits\limits_{1 \leq j \neq i \leq r}d_\text{avg}^{ij}
\end{align}
and $\delta \geq P_e^{*}$. As we can see, the length of the circuit is inversely proportional to the average distance between the classes. The closer they are, the longer the circuit will be needed in order to separate the states enough so they can be classified. 

It is worth emphasizing at this point that even if the condition of the above Theorem would permit a certain single-shot error probability, it is not necessarily the case (and usually not expected) that this single-shot error probability is actually achieved by the circuit.

\subsection*{Noisy circuits}

Realistic circuits are far from perfect. In the NISQ era, they are plagued by the eponymous noise that limits their performance. To study the impact of noise on the single-shotness of a quantum classifier, we study an error model where local depolarizing noise of survival probability $p$, \smash{$\calD_p[\rho] = p \rho + (1-p) \omega$}, $\omega$ being the maximally mixed state, is applied after every step of computation of the classifier to every qubit. 
To keep track of this mathematically, we will use the notation $\rho_t$ to denote the quantum state produced by the noisy embedding circuit after $t$ steps of computation. In our model of the embedding circuit, we have so far been counting the number of layers $L$, where each layer consists of a data embedding that we now assume to take one step of computation and fixed unitaries and variational gates that we assume to take $\ell$ steps of computation, yielding a total of $T = L(1 + \ell)$ steps of computation in our parametrized circuit.

A quantum circuit affected by local depolarizing noise will eventually produce a quantum state that is very close to the maximally mixed state $\omega$~\cite{stilck_franca2021limitations}. This result is established through a decay of the relative entropy between the circuit state after $t$ steps of computation and yields the following estimate~\cite{stilck_franca2021limitations}
\begin{align}
    D(\rho_t \fatpipe \omega) \leq p^{2t} D(\rho_0 \fatpipe \omega) \leq p^{2t} n,
\end{align}
where $n$ is the number of qubits. Combining this with Pinsker's inequality then gives us a bound on the trace distance to the maximally mixed state
\begin{align}
    \lVert \rho_t - \omega \rVert_1 \leq p^t \sqrt{2n},
\end{align}
which decreases exponentially in $t$. By a simple application of the triangle inequality, we obtain an estimate of the distance between two noisy embedded states
\begin{align}
    \lVert \rho_{\ttheta}(\xx) - \rho_{\ttheta}(\xx') \rVert_1 
    &\leq \lVert \rho_{\ttheta}(\xx) - \omega \rVert_1 + \lVert \rho_{\ttheta}(\xx') - \omega \rVert_1\\
    &\leq p^{L(1+\ell)} \sqrt{8n}\label{eqn:contraction_bound}.
\end{align}

This estimate, however, does not take into account the continuity of the circuits, so for close parameters it can actually be overly optimistic. As noisy channels can't increase the trace distance between states, we immediately see that the result of \cref{lem:continuity_of_parametrized_states} remains valid if noisy circuits are considered, albeit it will also be overly optimistic because it does not consider the noise in the system. A simple way of combining both estimates is to simply take the minimum of both, which gives
\begin{align}\label{eqn:min_of_two_bounds}
    \lVert \rho_{\ttheta}(\xx) - \rho_{\ttheta}(\xx') \rVert_1 \leq \min \{ p^{L(1+\ell)} \sqrt{8n}, L\Delta \lVert \xx - \xx' \rVert_1 \}.
\end{align}
We can improve on this estimate for shallow circuits. We perform a more involved analysis in \cref{appendix_noisy_continuity}, which implies the following Lemma.
\begin{lemma}\label{lem:noisy_continuity_estimate}
Consider a noisy embedding circuit with $L$ layers such that
\begin{align}
L \geq L_0 = \left\lceil 1 + \frac{1}{2(\ell+1)} \frac{\log 2n}{\log \frac{1}{p}} \right\rceil.
\end{align}
Then,
\begin{align}
\begin{split}
    &\lVert \rho_{\ttheta}(\xx) - \rho_{\ttheta}(\xx') \rVert_1  \\
    &\qquad\leq \Delta \lVert \xx - \xx' \rVert_1
    \left(L_0 + \sqrt{2n}\frac{p^{L_0(\ell+1)} - p^{L(\ell + 1)}}{1-p^{\ell + 1}}\right).\label{eqn:geometric_bound}
\end{split}
\end{align}
In the case $L \leq L_0$, the bound does not improve on the noiseless one.
\end{lemma}
Contrary to the noise-only estimate, the above bound does not decay to zero as $L \to \infty$, which means it can only improve on the estimate of \cref{eqn:min_of_two_bounds} in a limited parameter region. It does, however, always improve on the naive continuity estimate \ref{eqn:continuity_bound}. \cref{fig:noisy_estimates_plot} shows a comparison of the three different bounds. 

\begin{figure}
    \centering
    \includegraphics{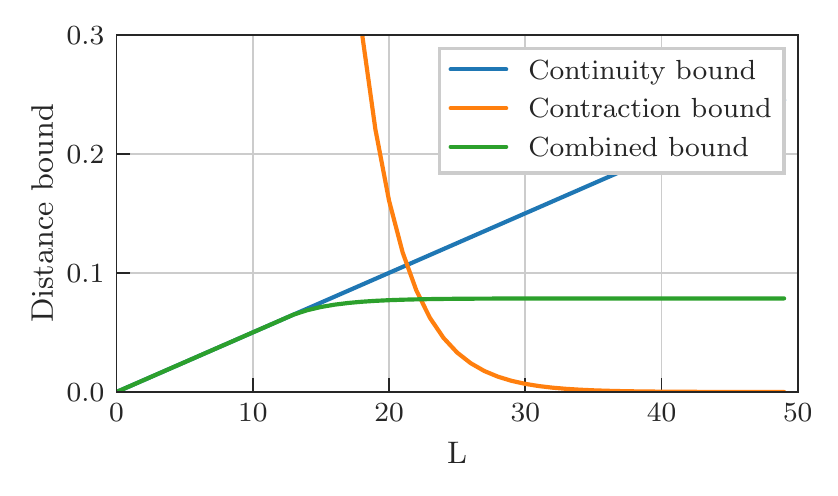}
    \caption{Comparison of the three different bounds for noisy circuits. We can distinguish (blue) the linearly growing continuity bound of \cref{eqn:continuity_bound}, (orange) the exponentially contracting bound of \cref{eqn:contraction_bound} and (green) the geometrically growing combined bound of \cref{eqn:geometric_bound}. The constants used for the plot are $n=1000$, $\lVert \xx - \xx' \rVert_1=0.005$, $\Delta=1$, $\ell=2$ and $p=0.9$.
    }
    \label{fig:noisy_estimates_plot}
\end{figure}

Now, we can connect back to the setting of single-shot classification in a similar way as to \cref{theorem_bayesian_bound}. The derivation can proceed analogously, only the constant $L$ is replaced with
\begin{align}
    L \rightarrow L_0 + \sqrt{2n}\frac{p^{L_0(\ell+1)} - p^{L(\ell + 1)}}{1-p^{\ell + 1}}.
\end{align}

Combining all three bounds we have derived for the noisy case gives us an upper bound on the achievable distance between the embedded average states per majority class predicted by the quantum classifier and as such allows us to analyze the possibility of single-shot classification by noisy quantum classifiers.

\section{Generically accurate single-shot models are hard to learn} \label{sect:hardness}

Having established a theoretical foundation of single-shot quantum machine learning, we have learned that single-shotness is a property of a trained model that is not necessarily a given. It depends on the specific way a quantum learning model embeds classical data into quantum states, and only materializes if quantum states of the same predicted label cluster together -- or in other words, they are not too close to states of a different predicted label. 

From this intuition, we see that it should in theory be possible to construct an embedding that enables single-shot classification for all possible labelings of a dataset. This is achieved by embedding all sufficiently separated inputs into mutually orthogonal directions of Hilbert space. One expect that, for example, we can construct such an embedding by using a large number of encoding gates interleaved with deep layers of random unitary gates. However, it is evident that such an embedding would result in poor generalization performance for any POVM inferred from the training data. Previously unseen datapoints would be embedded into the orthogonal complement of the embedded training data in Hilbert space. There, we cannot infer a good classification and the best strategy would be random guessing.

This \enquote{thought experiment} highlights an inherent problem of (quantum) learning models, namely that there is a fine trade-off between expressivity, trainability and generalization performance, which is the subject of the ongoing research in quantum machine learning~\cite{du2021learnability,caro2022generalization,du2023problem-dependent,gil-fuster2024understanding}. In this section, we will show that this trade-off also influences the single-shot property of a learning model.

To this end, we consider models that can be considered as \emph{generically accurate single-shot models}. These are models built from an embedding $\rho(\xx)$ such that for any possible labeling of the data, there exists a POVM $\{ \Pi_y \}_{y\in\calY}$ that accurately classifies in a single-shot way. 
Note that contrary to the preceding sections, where we treated single-shotness in a way that is completely detached from accuracy, here we additionally take it into account. 
For this notion to be sensible, we need to assume that all the datapoints fed into the model can't be arbitrarily close to each other. Formally, we assume that all possible datapoints encountered by the model have a lower bound on their distance in some norm: $\lVert \xx - \xx' \rVert > \xi$. Intuitively, we expect that a \emph{generically accurate single-shot model} should map all points that are separated by at least $\xi$ to (approximately) orthogonal directions in Hilbert space.

To study this specific class of models, we rely on the results of Ref.~\cite{cheng2015learnability}. There, the authors study the essentially similar problem of learning a POVM effect $\Pi$ from examples of the form 
\begin{align}
    \calD_m = \{ ( \rho_i, \Tr[ \rho_i \Pi] ) \}_{i=1}^m,
\end{align}
where the authors emphasize that the sample complexity does not change significantly when single measurement outcomes are given instead of expectation values. To connect this to our setting we consider the setting of binary classification. A classical datapoint $\xx$ can be seen as a classical description of $\rho(\xx)$, which means we already know the first half of the data tuples. The second half is then obtained by realizing that $\Tr[\rho(\xx) \Pi]$ can be seen as an instance of a label that should be either $0$ or $1$. As we assumed that there actually exists a POVM that realizes the desired classification, we can just assume that the given label -- either 0 or 1 -- is the expected value of the targeted POVM effect $\Pi$ that realizes the correct classification. The fact that the classifier is single-shot then enforces that the expected value of the effect $\Pi$ is actually identical to the target label. Hence, the training data of the classification task
\begin{align}
    \calD_m = \{ (\xx_i, y_i) \}_{i=1}^m
\end{align}
can be seen as equivalent to the setup of Ref.~\cite{cheng2015learnability} in the case of a generically accurate single-shot model.

Now, given this equivalence, Ref.~\cite{cheng2015learnability} posits that at least
\begin{align}
    m \geq O\left( \frac{D}{\epsilon^2}\right)
\end{align}
samples are necessary to learn a POVM effect on a $D$-dimensional Hilbert space to precision $\epsilon$ in operator norm. We are left to analyze what the dimension of the Hilbert space is that we need to consider. To  this end, we go back to our assumption that inputs that are $\xi$-apart in some norm should be mapped to different corners of the Hilbert space. This means we have that $D$ is approximately the number of balls in the norm of our choosing we need to cover the input space $\calX$. If the space has some sort of characteristic length $L$ and a dimension $d$, we expect to need
\begin{align}
    D \sim O\left( \left[\frac{L}{\xi}\right]^d  \right)
\end{align}
dimensions, which is inefficient in the size of the input as we expected from our thought experiment. 
This is also in line with the conclusions of Ref.~\cite{du2023problem-dependent}, which shows that overly expressive models have bad generalization performance.

\section{Conclusion} \label{sec:conclusion}

If quantum machine learning is ever to become useful, many obstacles have to be overcome. One of them is the \emph{measurement problem} that arises directly from the probabilistic nature of quantum mechanics. Predictions from quantum learning models are inherently random, and as such many executions of a quantum learning model have to be aggregated to obtain an actual prediction. 

In this work, we proposed the concept of \emph{single-shot} quantum machine learning. We established a rigorous definition of when a quantum classifier can circumvent the measurement problem by providing labels in a near-deterministic way. 

Exploiting a close connection between quantum classifiers and quantum multi-hypothesis testing allowed us to establish that \emph{how} single-shot a quantum classifier can be is fundamentally constrained by the distinguishability of the embedded quantum states. Considering that embeddings are usually realized by quantum circuits, we additionally showed how a certain depth of these circuits is necessary to realize a single-shot classifier.

Finally, we additionally established that a learning model cannot be single-shot in a generic way, i.e.\ for all possible labelings of the input. In that case, it would become too expressive and there would be no meaningful way for such a model to generalize well. 

This work can only be considered the start of our way to overcome the measurement problem. It invites future research in many directions. First of all, we only considered classification tasks in this work. A generalization to regression would surely also be possible and similar fundamental bounds could be expressed in terms of quantum metrology reductions instead of multi-hypothesis testing~\cite{meyer2023quantum}. 

Beyond extending the definitions and fundamental treatment of single-shot quantum machine learning, it is an immediately pressing question to understand how we can construct models that are single-shot? How can we enforce this property during training to obtain accuracy and single-shotness? A possible inspiration in this direction is to consider ways of training that are immediately related to the distinguishability of the embedded quantum states~\cite{lloyd2020quantum}. Can we also find explicit constructions for learning models or embedding circuits that enforce a certain degree of single-shotness?

\section{Acknowledgements}
The authors wish to thank Matthias Caro for his support in obtaining the continuity bounds used in this work.
This work has been supported by the DFG (CRC 183, FOR 2724), by the BMBF (Hybrid), the BMWK (PlanQK, EniQmA), the Munich Quantum Valley (K-8), QuantERA (HQCC) and the Einstein Foundation (Einstein Research Unit on Quantum Devices). This work has also been funded by the DFG under Germany's Excellence Strategy – The Berlin Mathematics Research Center MATH+ (EXC-2046/1, project ID: 390685689). It is further supported by the Government of Spain (Severo Ochoa CEX2019-000910-S, FUNQIP and European Union NextGenerationEU PRTR-C17.I1), Fundació Cellex, Fundació Mir-Puig, Generalitat de Catalunya (CERCA program) and European Union (PASQuanS2.1, 101113690). ER is a fellow of Eurecat's \enquote{Vicente López} PhD grant program.

\bibliography{main}

\clearpage

\appendix
\onecolumngrid

\section{Multi-hypothesis testing to binary reduction}\label{appendix_multi_ht_to_binary}

This section is devoted to the proof of the reduction from multi-hypothesis testing to binary hypothesis testing presented in the main text.

\begin{proof}[Proof of \cref{lemma_reduction}]
When facing a multi-hypothesis testing problem, additional side information can only improve our performance. Imagine now that a state $\rho_i$ is selected with probability $p_i$. A third party then selects the label $j$ of the other state in the set that is most difficult to distinguish from $\rho_i$ given the prior probabilities, i.e. 
\begin{align}
j = \operatornamewithlimits{argmax}_{1 \leq j \leq r} P_e^{*}(p_i\rho_i, p_j \rho_j).
\end{align}
The third party then gives us the information that the state we need to detect is either $\rho_i$ or $\rho_j$. The best strategy we can then perform is the optimal strategy for the binary discrimination of $\rho_i$ and $\rho_j$, leading to the bound
\begin{align}
    \begin{split}
    P_e^{*}&( \{p_i \rho_i \}_{i=1}^r )\\
    &\geq \sum_{i=1}^r p_i \max_{1 \leq j \neq i \leq r} P_e^{*}\left(\frac{p_i}{p_i + p_j} \rho_i, \frac{p_j}{p_i + p_j} \rho_j\right) 
    \end{split}\\
    &\geq \min_{1 \leq i \leq r} \max_{1 \leq j \neq i \leq r} P_e^{*}\left(\frac{p_i}{p_i + p_j} \rho_i, \frac{p_j}{p_i + p_j} \rho_j\right), 
\end{align}
where we divided the prior probabilities in the error probability term by $p_i + p_j$ to account for the fact that we conditioned on the fact that it is either $p_i$ or $p_j$ and the $p_i$ in front is the probability of obtaining $i$ in the first place. We lower-bound the convex combination by the smallest of the summed-up terms to arrive at the lemma. 
\end{proof}

\section{Continuity bounds for variational quantum circuits}\label{appendix_continuity}

This section is devoted to establishing the desired continuity estimates for quantum embeddings. Let us first rigorously define what we understand as a parametrized quantum circuit.
\begin{definition}[Parametrized quantum circuit]\label{def:pqc}
A parametrized quantum circuit is a unitary transformation obtained by successively applying gates on a quantum computer, where the parametrized gates are time evolutions under self-adjoint generating Hamiltonians $G_i$:
\begin{align}
    U(\ttheta) = \prod_{i=1}^P e^{-i \theta_i G_i} V_i,
\end{align}
the non-parametrized gates $V_i$ can be arbitrary. We denote the set of all generators as $\calG$.
\end{definition}

We can then establish the following fundamental continuity estimate for parametrized quantum circuits in diamond norm.
\begin{theorem}[Diamond norm continuity]\label{thm:diamond_norm_continuity}
Let $U(\ttheta)$ denote a parametrized quantum circuit with set of generators $\calG$ as in \cref{def:pqc} and $\calU(\ttheta)[\rho] = U(\ttheta) \rho U^{\dagger}(\ttheta)$ its associated channel. Then, the channel obeys the following continuity bound:
\begin{align}
    \lVert \calU(\ttheta)^{\otimes m} - \calU(\ttheta')^{\otimes m}\rVert_{\diamond} \leq m \Delta \lVert \ttheta - \ttheta' \rVert_1,
\end{align}
where $\Delta$ is the largest spectral spread of all generating Hamiltonians,
\begin{align}
    \Delta  = \max  \{ \lambda_{\max}(G) - \lambda_{\min}(G) \, | \, G \in \calG \}.
\end{align}
There exists circuit that saturate the above bound.
\end{theorem}
\begin{proof}
We first establish continuity of a single parametrized gate in operator norm. Observe that
\begin{align}
    \lVert R_G(\theta) - R_G(\theta') \rVert_{\infty} &= \lVert e^{-i \theta G} - e^{-i \theta' G} \rVert_{\infty} \\
    &= \lVert e^{-i (\theta - \theta') G} - \bbI \rVert_{\infty} \\
    &= \sup_{\lambda \in \operatorname{spec}(H)} |\exp(-i \lambda (\theta - \theta') ) - 1 | \\
    &\leq \sup_{\lambda \in \operatorname{spec}(H)} |\lambda(\theta - \theta') | \\
    &\leq \lVert G \rVert_{\infty} |\theta - \theta'|,
\end{align}
where the first equality is the definition of the rotation gate, the second equality follows from unitary invariance of the operator norm, the third equality from the fact that the identity commutes with the unitary and the first inequality from the upper bound $|\exp(-i x) - 1| \leq |x|$ for $x\in\mathbb{R}$.

We can use the freedom of adding a global phase, which corresponds to adding a multiple of the identity to the generator, to optimize this bound. The minimal operator norm is 
\begin{align}
    \min_{r \in \bbR} \lVert G + r \bbI \rVert_{\infty} = \frac{\lambda_{\max}(G) - \lambda_{\min}(G)}{2} \eqqcolon \frac{\Delta(G)}{2},
\end{align}
where we defined the spectral spread $\Delta(H)$.

This reasoning easily expands to the full circuit $U(\ttheta)$, which can be seen by noting that a straightforward application of the triangle inequality and the unitary invariance gives us a way to \enquote{chop away} gates one after another:
\begin{align}
    \lVert U R_G(\theta) - U' R_G(\theta') \rVert_{\infty} &= \lVert U R_G(\theta - \theta') - U' \rVert_{\infty} \\
    &= \lVert U R_G(\theta - \theta') - U' + U - U \rVert_{\infty} \\
    &\leq \lVert U R_G(\theta - \theta') - U \rVert_{\infty} + \lVert U - U' \rVert_{\infty} \\
    &= \lVert  R_G(\theta - \theta') - \bbI \rVert_{\infty} + \lVert U - U' \rVert_{\infty} \\
    &\leq \frac{\Delta(G)}{2} |\theta - \theta'| + \lVert U - U' \rVert_{\infty}.
\end{align}
Applying this relation recursively yields the bound
\begin{align}
    \lVert U(\ttheta) - U(\ttheta') \rVert_{\infty} \leq \sum_{i=1}^P \frac{\Delta(G_i)}{2} |\theta_i - \theta_i'| \leq \left(\max_{i} \frac{\Delta(G_i)}{2}\right) \lVert \ttheta - \ttheta'\rVert_1.
\end{align}

Lemma 5 of Ref.~\cite{caro2022generalization} now allows us to use the operator norm estimate to bound the diamond norm:
\begin{align}
    \lVert \calU(\ttheta) - \calU(\ttheta') \rVert_{\diamond} 
    \leq 2 \lVert U(\ttheta) - U(\ttheta') \rVert_{\infty}
    \leq \left(\max_{i} \Delta(G_i)\right) \lVert \ttheta - \ttheta'\rVert_1.
\end{align}

To conclude the bound claimed in the theorem statement we make use of some fundamental properties of the diamond norm (see e.g.\ Proposition 3.48 of Ref.~\cite{watrous2018theory}) which we can use to treat tensor products of channels:
\begin{align}
    \lVert \calN \otimes \calM - \calN' \otimes \calM' \rVert_{\diamond}
    &\leq \lVert \calN \otimes \bbI - \calN' \otimes \bbI \rVert_{\diamond} + \lVert \bbI \otimes \calM - \bbI \otimes \calM' \rVert_{\diamond} \\
    &= \lVert \calN - \calN' \rVert_{\diamond} +  \lVert \calM - \calM' \rVert_{\diamond}.
\end{align}
The above implies that
\begin{align}
    \lVert \calN^{\otimes m} - \calM^{\otimes m} \rVert_{\diamond} \leq m \lVert \calN - \calM \rVert_{\diamond},
\end{align}
which indeed implies the desired statement.

The tightness of the bound is established by considering the saturating example of parallel $Z$ rotations on $n$ qubits:
\begin{align}
    U(\ttheta) = \bigotimes_{i=1}^n R_{Z}(\theta_i).
\end{align}
We will use the following fact for self-adjoint and unitary generators
\begin{align}
    R_Z(\theta) = \bbI \cos \theta -i Z \sin \theta
\end{align}
We consider an infinitesimal perturbation around $\ttheta =0$, for which we have that
\begin{align}
    \bigotimes_{j=1}^n R_{Z}(\diff \theta_j) &= \bigotimes_{j=1}^n \left( \bbI \cos \diff \theta_j -i Z \sin \diff \theta_j\right) \\
    &= \bigotimes_{j=1}^n \left( \bbI -i Z \diff \theta_j\right)\\
    &= \bigotimes_{j=1}^n \bbI - \sum_{j=1}^n i Z_j \diff \theta_j + O(\lVert \diff \ttheta \rVert_{1}^2),
\end{align}
where we defined $Z_j = \bbI^{\otimes j-1} \otimes Z \otimes \bbI^{\otimes n - j}$. With this at hand, we see that
\begin{align}
    \lVert U(0) - U(\diff \ttheta) \rVert_{\infty} &= \Big\lVert \bbI - \bigotimes_{j=1}^n R_{Z}(\diff \theta_j) \Big\rVert_{\infty} 
    = \Big\lVert \sum_{j=1}^n Z_j \diff \theta_j \Big\rVert_{\infty}  + O(\lVert \diff \ttheta \rVert_{1}^2)
\end{align}
As $\sum_{i=1}^n Z_i \diff \theta_i$ has terms acting on independent tensor factors, we have a simple structure for the eigenstates which we can understand as labeled by bitstrings $\bb$, such that
\begin{align}
    \sum_{j=1}^n Z_j \diff \theta_j \ket{\bb} = \sum_{j=1}^n (-1)^{b_j} \diff \theta_j \ket{\bb}.
\end{align}
We see that the eigenvalues are all possible signed sums of the $\diff \theta_j$, which means that
\begin{align}
    \Big\lVert \sum_{j=1}^n Z_j \diff \theta_j \Big\rVert_{\infty} &= \max_{\bb}\sum_{j=1}^n (-1)^{b_j} \diff \theta_j 
    = \lVert \diff \ttheta \rVert_1.
\end{align}
Therefore,
\begin{align}
    \lVert \calU(0) - \calU(\diff \ttheta) \rVert_{\diamond} &= 2 \lVert U(0) - U(\diff \ttheta) \rVert_{\infty} 
    = 2 \lVert \diff \ttheta \rVert_1 +  O(\lVert \diff \ttheta \rVert^2),
\end{align}
which -- noting that $\Delta(Z) = 2$ -- saturates the bound as claimed. We now showed the saturation specifically for $m=1$, but a similar construction works for general $m$ by just repeating the above construction in parallel.
\end{proof}

We can now use this result to prove \cref{lem:continuity_of_parametrized_states} of the main text.
\begin{proof}[Proof of \cref{lem:continuity_of_parametrized_states}]
In our setup, we encounter a parametrized quantum state with two kinds of parameters: the data $\xx$ and the variational parameters $\ttheta$. As the variational parameters are equal for the two circuits we compare, we can disregard them and absorb them into the fixed unitaries of the circuit. The encoding of the data is repeated $L$ times in our example. This we can take care of by viewing each of these instances as different parameters $\{ \tilde{x}_j \}_{j=1}^{Ld}$ and noting that $\lVert \tilde\xx \rVert_1 = L \lVert \xx \rVert_1$. Hence, the diamond norm continuity of \cref{thm:diamond_norm_continuity} implies that
\begin{align}
    \lVert \calU_{\ttheta}(\xx) - \calU_{\ttheta}(\xx') \rVert_{\diamond} \leq L \Delta \lVert \xx - \xx' \rVert_1,
\end{align}
where the maximal spectral spread $\Delta$ is now calculated only over the gates encoding the parameters.
We can conclude the bounds on the output states through a chain of straightforward inequalities
\begin{align}
    \lVert \rho_{\ttheta}(\xx) - \rho_{\ttheta}(\xx') \rVert_1 &= \lVert \calU_{\ttheta}(\xx)[\rho_0] - \calU_{\ttheta}(\xx')[\rho_0] \rVert_1 \\
    &\leq  \lVert \calU_{\ttheta}(\xx) - \calU_{\ttheta}(\xx') \rVert_{\diamond} \lVert \rho_0 \rVert_1 \\
    &\leq L \Delta \lVert \xx - \xx' \rVert_1.
\end{align}
\end{proof}

\section{Proof of \cref{theorem_bayesian_bound} -- continuity bound for the noiseless case}\label{proof_continuity_noiseless}

Because we reduced the task of single-shot classification to multi-hypothesis testing in \cref{thm:bayesian_error_probability_lower_bound}, it suffices to establish a lower bound on the error probability of the multi-hypothesis test between the average states of the majority labels.

Putting together \cref{lemma_reduction} and \cref{helstrom} we find:
\begin{align}
    P_e^{*} & \geq \min_{1 \leq i \leq r} \max_{1 \leq j \neq i \leq r} P_e^{*}\left(\frac{p_i}{p_i + p_j} \rho_i, \frac{p_j}{p_i + p_j} \rho_j\right)\\
    & = \min_{1 \leq i \leq r} \max_{1 \leq j \neq i \leq r} \frac{1}{2} \left(1 - \bigg\| \frac{p_i}{p_i + p_j} \rho_i - \frac{p_j}{p_i + p_j} \rho_j \bigg\|_1 \right).
\end{align}
Analyzing the elements inside the trace distance and renaming
\begin{align}
p_+ = \frac{p_i}{p_i + p_j}, \;\; \rho_+ = \rho_i \;\; \calX_+ = \calX_i\\
p_- = \frac{p_j}{p_i + p_j}, \;\; \rho_- = \rho_j \;\; \calX_- = \calX_j,
\end{align}
we see that
\begin{align}
    \|p_+\rho_+ - p_-\rho_-\|_1&=\bigg\|\int_{\mathcal{X_+}} \diff p(\xx_+) \, \rho(\xx_+) - \int_{\mathcal{X_-}} \diff p (\xx_-) \, \rho(\xx_-) \bigg\|_1\\
    &= \bigg\|\int_{\mathcal{X_+}} \diff p(\xx_+) \, \rho(\xx_+) \int_{\mathcal{X_-}}\frac{\diff p(\xx_-)}{p_-} - \int_{\mathcal{X_-}}  \diff p(\xx_-) \rho(\xx_-) \, \int_{\mathcal{X_+}}\frac{\diff p(\xx_+)}{p_+}\bigg\|_1\\
    &= \bigg\|\frac{1}{p_+p_-}\int_{\mathcal{X_+}} \diff p(\xx_+) \int_{\mathcal{X_-}}\diff p(\xx_-) \, \left(p_+\rho(\xx_+)-p_-\rho(\xx_-)\right) \bigg\|_1\\
    &\leq \frac{1}{p_+p_-}\int_{\mathcal{X_+}}dp(\xx_+) \int_{\mathcal{X_-}}dp(\xx_-)\|p_+\rho(\xx_+)-p_-\rho(\xx_-) \|_1.
\end{align}
Where we have applied the triangle inequality. Looking further into the one-norm:
\begin{align}
    \|p_+\rho(\xx_+)-p_-\rho(\xx_-) \|_1 
    &= \|p_+\rho(\xx_+)-p_-\rho(\xx_-) -p_-\rho(\xx_+) + p_-\rho(\xx_+)\|_1 \\
    &= \|(p_+-p_-)\rho(\xx_+)+p_-(\rho(\xx_+)-\rho(\xx_-))\|_1 \\
    &\leq |p_+-p_-|+p_-\|\rho(\xx_+)-\rho(\xx_-)\|_1,
\end{align}
where we have again applied the triangle inequality. Notice that we could have used $p_+$ instead of $p_-$ in the added terms, therefore:
\begin{equation}
    \|p_+\rho(\xx_+)-p_-\rho(\xx_-) \|_1 \leq |p_+-p_-|+\min(p_+,p_-)\|\rho(\xx_+)-\rho(\xx_-)\|_1.
\end{equation}
Now everything is ready to apply the bound from \cref{lem:continuity_of_parametrized_states}:
\begin{equation}
    \|\rho(\xx_+)-\rho(\xx_-) \|_1 \leq L\Delta \|\xx_+-\xx_-\|_1.
\end{equation}
Putting everything together yields
\begin{align}
    P_e^{*} &\geq \min_{1 \leq i \leq r} \max_{1 \leq j \neq i \leq r} \frac{1}{2} \left(1 - \frac{|p_i-p_j|}{p_i + p_j}-L\Delta\frac{\min(p_i,p_j)}{p_i + p_j}\int_{\mathcal{X_+}}\frac{dp(\xx_+)}{p_+} \int_{\mathcal{X_-}}\frac{dp(\xx_-)}{p_-} \|\xx_+-\xx_-\|_1\right)\\
    &= \min_{1 \leq i \leq r} \max_{1 \leq j \neq i \leq r} \frac{1}{2} \left(1 - \frac{|p_i-p_j|}{p_i + p_j}-L\Delta\frac{\min(p_i,p_j)}{p_i + p_j}d_\text{avg}^{ij}\right)\\
    &= \min_{1 \leq i \leq r} \max_{1 \leq j \neq i \leq r} \frac{\min(p_i,p_j)}{p_i + p_j}\left(1 -\frac{L\Delta d_\text{avg}^{ij}}{2}\right)
\end{align}
In the above derivation, we used the following identity:
\begin{align}
    1 - \frac{|p_i-p_j|}{p_i + p_j}
    &= \begin{cases}
      \cfrac{p_i + p_j-(p_i-p_j)}{p_i + p_j} & \text{if } p_i > p_j\\
      \cfrac{p_i + p_j-(p_j-p_i)}{p_i + p_j} & \text{if } p_i < p_j
    \end{cases} \\
    &= \begin{cases}
      \cfrac{2p_j}{p_i + p_j} & \text{if } p_i > p_j\\
      \cfrac{2p_i}{p_i + p_j} & \text{if } p_i < p_j
    \end{cases} \\
    &= \frac{2\min(p_i,p_j)}{p_i + p_j}.
\end{align}

\section{Proof of the noisy continuity bound}\label{appendix_noisy_continuity}

We will derive a result that will have \cref{lem:noisy_continuity_estimate} as a corollary.
\begin{theorem}
Let 
\begin{align}
    \rho_t &= \calD_p^{\otimes n} \calU_t \calD_p^{\otimes n} \calU_{t-1} \dots \calU_1 \calD_p^{\otimes n} [\rho_0] \\
    \sigma_t &= \calD_p^{\otimes n} \calV_t \calD_p^{\otimes n} \calV_{t-1} \dots \calV_1 \calD_p^{\otimes n} [\rho_0]
\end{align}
be the outputs of two quantum circuits of depth $t$ such that $\lVert \calU_i - \calV_i \rVert_{\diamond} \leq \Gamma_i \leq \Gamma_{\max}$. Then,
\begin{align}
    \lVert \rho_t - \sigma_t \rVert_1 &\leq \Gamma_{\max} \times \begin{cases} t & \text{if } t \leq t_0 \\
    t_0 + \sqrt{2n}\frac{p^{t_0} - p^t}{1-p}&  \text{if } t > t_0
    \end{cases},
\end{align}
where 
\begin{align}
    t_0 = \left\lceil 1+ \frac{1}{2} \frac{\log 2n}{\log \frac{1}{p}} \right\rceil.
\end{align}
\end{theorem}
\begin{proof}
We will use the same strategy that establishes the continuity bound of \cref{thm:diamond_norm_continuity}, however, in the course of the derivation we will make use of the fact that the maximally mixed state is a fixed point of any unitary, and hence $(\calV^{\dagger} \calU - \bbI)[\omega] = 0$ to improve our estimate.
\begin{align}
    \lVert \rho_t - \sigma_t \rVert_1 &= \lVert \calD_p^{\otimes n} (\calU_t[\rho_{t-1}] - \calV_t[\sigma_{t-1}])\rVert_1 \\
    &\leq \lVert \calU_t[\rho_{t-1}] - \calV_t[\sigma_{t-1}]\rVert_1 \\
    &= \lVert \calV_t^{\dagger}\calU_t[\rho_{t-1}] - \sigma_{t-1} \rVert_1\\
    &= \lVert (\calV_t^{\dagger}\calU_t - \bbI)[\rho_{t-1}] + \rho_{t-1} - \sigma_{t-1} \rVert_1 \\
    &= \lVert (\calV_t^{\dagger}\calU_t - \bbI)[\rho_{t-1} - \omega] + \rho_{t-1} - \sigma_{t-1} \rVert_1 \\
    &\leq \lVert \calV_t^{\dagger}\calU_t - \bbI \rVert_{\diamond} \min\{ \lVert \rho_{t-1} \rVert_1, \lVert \rho_{t-1} - \omega \rVert_1 \} + \lVert \rho_{t-1} - \sigma_{t-1} \rVert_1\\
    &\leq \Gamma_t \min\{ 1, p^{t-1} \sqrt{2n}\} + \lVert \rho_{t-1} - \sigma_{t-1} \rVert_1\\
    &\leq \sum_{i=1}^t \Gamma_i \min\{ 1, p^{i-1} \sqrt{2n}\} \\
    &\leq \Gamma_{\max}  \sum_{i=1}^t \min\{ 1, p^{i-1} \sqrt{2n}\} 
\end{align}
To complete our proof, note that we have a geometric series for the terms that are powers in $p$ which gives
\begin{align}
    \sum_{i=1}^t p^{i-1} = \frac{1 - p^t}{1-p}.
\end{align}
As the terms $p^{i-1}$ are monotonically decreasing in $i$, the noisy estimate will be better as soon as $i \geq t_0$ where $t_0$ solves
\begin{align}
    1 \geq p^{t_0-1}\sqrt{2n} \ \Leftrightarrow \ t_0 = \left\lceil 1+ \frac{1}{2} \frac{\log 2n}{\log \frac{1}{p}} \right\rceil.
\end{align}
The final bound is thus an improved version of the continuity version that saturates asymptotically:
\begin{align}
    \lVert \rho_t - \sigma_t \rVert_1 &\leq \Gamma_{\max} \times \begin{cases} t & \text{if } t \leq t_0 \\
    t_0 + \sqrt{2n}\left(\frac{1 - p^t}{1-p} - \frac{1 - p^{t_0}}{1-p}\right) &  \text{if } t > t_0
    \end{cases} \\
    &= \Gamma_{\max} \times \begin{cases} t & \text{if } t \leq t_0 \\
    t_0 + \sqrt{2n}\frac{p^{t_0} - p^t}{1-p}&  \text{if } t > t_0
    \end{cases}
\end{align}
\end{proof}
We can now specialize the above Theorem to the setting of the main text.
\begin{proof}[Proof of \cref{lem:noisy_continuity_estimate}]
    In the setting of \cref{lem:noisy_continuity_estimate}, the total circuit depth is $t = L(\ell +1)$, but only $L$ of the computational steps where data is encoded have a non-zero upper bound on the diamond norm distance given by $\Gamma_{\max} \leq \Delta \lVert \xx - \xx' \rVert_1$. We therefore obtain an estimate of the form
    \begin{align}
        \lVert \rho_{\ttheta}(\xx) - \rho_{\ttheta}(\xx') \rVert_1 \leq 
        \Gamma_{\max} \sum_{i=1}^L \min\{ 1, p^{\ell(i-1)} \sqrt{2n} \}, 
    \end{align}
    where $i$ now counts the number of layers. The geometric series then takes the form
    \begin{align}
        \sum_{i=1}^L p^{(i-1)(\ell+1)}  = \sum_{i=1}^L (p^{\ell +1})^{i-1} = \frac{1 - p^{L (\ell + 1)}}{1 - p^{\ell + 1}}.
    \end{align}
    and we need to solve the number of layers $L_0$ for
    \begin{align}
        1 \geq p^{(L_0 -1)(\ell +1) }\sqrt{2n} \ \Leftrightarrow \ L_0 = \left\lceil 1 + \frac{1}{2(\ell+1)} \frac{\log 2n}{\log \frac{1}{p}} \right\rceil.
    \end{align}
    Combining everything, we obtain the stated bound.
\end{proof}

\end{document}